\documentclass[reqno]{amsart}
\usepackage[top=1.2in,bottom=1.2in,left=1.3in,right=1.3in]{geometry}
\usepackage{amscd}
\usepackage{amssymb}
\usepackage{calc,pifont}
\usepackage[T1]{fontenc}
\usepackage[latin1]{inputenc}

\usepackage[arrow,curve,matrix]{xy}
\usepackage{times}
\usepackage{graphicx}
\usepackage{extarrows}
\usepackage{flafter}
\usepackage{placeins}
\usepackage{enumitem}
\usepackage[OT2,T1]{fontenc}
\DeclareSymbolFont{cyrletters}{OT2}{wncyr}{m}{n}
\DeclareMathSymbol{\Sha}{\mathalpha}{cyrletters}{"58}

\usepackage{centernot}
\usepackage{pb-diagram}
\usepackage{mathrsfs}
\usepackage[color]{showkeys}

\definecolor{refkey}{rgb}{1,1,1}
\definecolor{labelkey}{rgb}{1,1,1}
\definecolor{cite}{rgb}{0.9451,0.2706,0.4941}
\definecolor{ruri}{rgb}{0.0078,0.4022,0.8010}

\usepackage[%
bookmarks=true,bookmarksnumbered=true,%
colorlinks=true,linkcolor=ruri,citecolor=red%
 ]{hyperref}

\def\rank{{\rm rank}}
\def\rk{{\rm rk}}

\def\End{{\rm End}}

\def\ext{{\rm Ext}}

\theoremstyle{plain}

\newtheorem{theorem}{Theorem}[section]

\newtheorem{proposition/example}[theorem]{Proposition/Example}
\newtheorem{proposition}[theorem]{Proposition}
\newtheorem{corollary}[theorem]{Corollary}

\theoremstyle{definition}
\newtheorem{definition}[theorem]{Definition}

\newtheorem{example}[theorem]{Example}

\newtheorem{conjecture/question}[theorem]{Conjecture/Question}

\newtheorem{remark/definition}[theorem]{Remark/Definition}
\newtheorem{definition/notation}[theorem]{Definition/Notation}

 \theoremstyle{remark}

\numberwithin{equation}{section}

\usepackage{marginnote,color,changepage}
\usepackage[usenames,dvipsnames]{xcolor}

\def\shownotes{\def\inline##1##2##3{ \begin{adjustwidth}{3mm}{7mm}\mbox{}\par \noindent
{\color{##1}\hspace{-1.9cm}{\large ##2}\vspace{-\baselineskip}\\##3}
\newline\end{adjustwidth}} \def\inlinewide##1##2##3{ \begin{adjustwidth}{0mm}{0cm}\mbox{}\par \noindent
{\color{##1}\hspace{-1.6cm}{\large ##2}\vspace{-\baselineskip}\\##3}
\newline\end{adjustwidth}}  \def\marg##1##2##3{\marginnote{\color{##1}{\large ##2}\\{\small ##3}}[-.8cm]}}

\shownotes

\begin{document}
\title {\textbf{  Degenerate and Stable Yang-Mills-Higgs Pairs}}
\author{Zhi Hu \& Sen Hu}
\address{School of Mathematics, University of Science and
Technology of China\\Hefei, 230026, China}
 \email{{\tt
halfash@mail.ustc.edu.cn,shu@ustc.edu.cn}}

\maketitle
 \begin{abstract}In this paper, we introduce some  notions on the pair consisting of a Chern connection  and a Higgs field closely related to the first and second variation of Yang-Mills-Higgs functional, such as strong Yang-Mills-Higgs pair, degenerate Yang-Mills-Higgs pair, stable Yang-Mills-Higgs pair. We investigate   some properties of such pairs.
\end{abstract}
\tableofcontents
\section{Introduction }
Since 1950s, Yang-Mills theory first explored by several physicists had a profound impact on the developments of
differential and algebraic geometry. A remarkable fruit owed to  Donaldson is  constructing invariants of 4-manifolds via studying the homology of the moduli space of anti-self-dual $SU(2)$-connections, where technical challenges come from Uhlenbeck compactification of moduli space and handling  singularities through the metric perturbations\cite{1,2}.
In 1987 Hitchin  considered the 2-dimensional reduction of the self-dual Yang-Mills equations on  $\mathbb{R}^4$ as a manner of symmetry breaking, then he introduced a (1,0)-form valued in  adjoint vector bundle  $\phi$ called the  Higgs field for  the  Riemann surface, which is  described by the so-called Hitchin self-duality equations\cite{3}:
\begin{align*}
  F_A+[\phi,\overline{\phi}]&=0,\\
  d_A^{\prime\prime}\phi&=0.
  \end{align*}
Influenced by the Hitchin's work, Simpson generalized the conception of Higgs field to the higher dimensional case\cite{31}, and he made
great innovations in various areas of algebraic geometry\cite{4,s2,s3}.  Since then Higgs bundles have emerged
in the last two decades as a central object of study in geometry,
with several links to physics and number theory.

Let us first recall some basic definitions.
\begin{definition}(\cite{b,5,s})
Let $X$ be an $n$-dimensional compact K\"{a}hler
manifold with K\"{a}hler form $\omega$, and let
$\Omega^1_X$ be the
the sheaf of holomorphic 1-forms on $X$. A Higgs sheaf over $X$ is a coherent
sheaf $E$ of dimension $n$ over $X$, together with a morphism $\phi:E\rightarrow E\otimes\Omega^1_X$
of $\mathcal{O}_X$-modules (that is
usually called the Higgs field), such that the morphism $\phi\wedge\phi:E\rightarrow E\otimes\Omega^2_X$
 vanishes. A Higgs
bundle  is a locally-free Higgs sheaf. A subsheaf $F$ of $E$ is called the Higgs
subsheaf if $\phi(F)\subset F\otimes\Omega^1_X$, i.e. the pair
$F = (F, \phi|_F )$ becomes itself a Higgs sheaf. Let $(E_1,\phi_1)$ and $(E_2,\phi_2)$ be two Higgs sheaves over $X$. A
morphism between them is a map $E_1\rightarrow E_2$  such that the following diagram commutes
$$\CD
  E_1 @>\phi_1>> E_1\otimes\Omega^1_X\\
  @V f VV @V f\otimes 1 VV  \\
  E_2 @>\phi_2>> E_1\otimes\Omega^1_X.
\endCD$$
 A Higgs sheaf $(E,\phi)$  over $X$
is called $\omega$-stable (resp. $\omega$-semistable) if it is torsion-free, and for any Higgs subsheaf $F$ with $0<\rank(F)<\rank(E)$  the
inequality $\mu_\omega(F)<\mu_\omega(E)$ (resp. $\leq$) holds, where the slop $\mu_\omega(E)$ of $E$ is defined by $\mu_\omega(E):=\frac{\deg_\omega(E)}{\rank(E)}=\frac{\int_X c_1(E)\wedge \omega^{n-1}}{\rank(E)}$. We say that a $\omega$-semistable Higgs sheaf is $\omega$-polystable if it
decomposes into a direct sum of  $\omega$-stable Higgs sheaves.
\end{definition}

Here we briefly mention some significant results of Simpson\cite{4,s2,s3}  which will be used in this paper.  Let $(E,\phi)$ be a Higgs bundle over a compact  K\"{a}hler manifold $(X,\omega)$, then it is $\omega$-polystable if and only if  there exists a Hermitian-Yang-Mills-Higgs metric  on it. This result  as an extension of the Hitchin-Kobayashi correspondence for Higgs
bundles is even true for manifolds which are not necessarily compact (satisfying some additional analytic requirements). In particular, there is an equivalence of
categories between the category of semisimple
flat bundles  and the category
of polystable Higgs bundles with vanishing (first and second) Chern classes, both
being equivalent to the category of harmonic bundles. If $X$ is a smooth projective variety,
there is a moduli space $\mathcal{M}$ as a quasi-projective variety of polystable Higgs bundles with vanishing Chern class and there is a natural action of $\mathbb{C}^*$ on $\mathcal{M}$
via multiplying  the Higgs field by a non-zero complex number $c$. Simpson showed that the limit as $c$ goes to zero exists and is unique. The limit
is therefore a fixed point of the $\mathbb{C}^*$-action, and is made into the variation of polarized Hodge structures.

 A Yang-Mills-Higgs system is a collection of data $\{(E,\phi), h, d_A\}$ for  a Higgs bundle $(E,\phi)$ with  a Hermitian metric $h$  and the corresponding Chern connection  $d_A$ on $E$, where the pair $(A,\phi)$  is called a Hitchin pair. These data give rise to a non-metric
connection $\mathcal{D}_{(A,\phi)}$ called the
Hitchin-Simpson connection on $E$ in the following way: $\mathcal{D}_{(A,\phi)}=d_A+\phi+\phi^*$ where  $ \phi^*:E\rightarrow E\otimes\overline{\Omega^1_X}$ is the adjoint of the Higgs field with respect to the  Hermitian metric
$ h$, namely $h(\theta(Y)s,t)=h(s,\theta^*(\bar Y)t)$ for the complex tangent vector $Y$ and the sections $s, t$ of $E$.   The curvature $\mathcal{R}_{(A,\phi)}$ of the Hitchin-Simpson connection is given by
  $\mathcal{R}_{(A,\phi)}=\mathcal{D}_{(A,\phi)}^2=(d_A+\phi+\phi^*)\wedge(d_A+\phi+\phi^*)
   =F_A+[\phi,\phi^*]+d^\prime_A(\phi)+d^{\prime\prime}_A(\phi^*)$,
  where $F_A=d_A^2$ denotes the  curvature of the Chern connection $d_A$  decomposed  into (1,0)-part $d^\prime_A$ and (0,1)-part $d^{\prime\prime}_A$. Let $L:\Lambda^{p,q}\rightarrow\Lambda^{p+1,q+1}$ be the operator of the multiplication by the K\"{a}hler form $\omega$, $\Lambda$ be the adjoint of $L$, and denote the (1,0)-part and (0,1)-part of the Hitchin-Simpson connection $\mathcal{D}_{(A,\phi)}$ by $\mathcal{D}^\prime_{(A,\phi)}$ and $\mathcal{D}^{\prime\prime}_{(A,\phi)}$ respectively. The following K\"{a}hler identities can be easily checked:
  \begin{eqnarray}\label{eq:a}
  \begin{aligned}
    i[\Lambda, \mathcal{D}^\prime_{(A,\phi)}]&=-(\mathcal{D}^{\prime\prime}_{(A,\phi)})^*,\\
    i[\Lambda, \mathcal{D}^{\prime\prime}_{(A,\phi)}]&=(\mathcal{D}^{\prime}_{(A,\phi)})^*.
  \end{aligned}
  \end{eqnarray}
  Let $\Theta_{(A,\phi)}$ stand for the anti-Hermitian part of  $\mathcal{D}_{(A,\phi)}$, which is exactly the (1,1)-part $F_A+[\phi,\phi^*]$. We define the mean curvature $\mathcal{K}_{(A,\phi)}$ of the Hitchin-Simpson connection, just
by contraction of its curvature with the operator $i\Lambda$, i.e. $\mathcal{K}_{(A,\phi)}=i\Lambda\mathcal{R}_{(A,\phi)}=i\Lambda\Theta_{(A,\phi)}\in\textrm{End}(E)$ or in terms of local frame field $\{e_i\}_{i=1}^{r}$ of $E$ and local coordinates $\{z_\alpha\}_{\alpha=1}^n$ of $X$, $\mathcal{K}^i_j=\omega^{\alpha\bar\beta}\mathcal{R}^i_{j\alpha\bar\beta}$.

For a Yang-Mills-Higgs system, one may attempt to  solve the following nonlinear equation
 \begin{equation}\label{a1}
  \det(\mathcal{K}_{(A,\phi)})=\lambda
\end{equation}
  for a constant $\lambda$. If the solution of the equation above exists  the corresponding system is called a special  Yang-Mills-Higgs system, in particular, it is called a degenerate Yang-Mills-Higgs system for the case $\lambda=0$. When $h$ is a weak  Hermitian-Yang-Mills metric on a Higgs bundle $(E,\phi)$, that is it satisfies the equation $\mathcal{K}_{(A,\phi)}=f\cdot Id_E$
for  a real function $f$ defined on $X$, we can obtain a spectial Yang-Mills-Higgs system $\{(E,\phi),\tilde h, d_A\}$ by taking $\tilde h$ to be a conformal transformation of $h$ such that $\tilde h$ is a Hermitian-Yang-Mills metric. Indeed,
since there is a solution $u$ for the equation $\Delta u=c-f$ where $c$ is a constant determined by $c\int_X\omega^n=\int_Xf\omega^n$, we only need to let $\tilde h=e^uh$ which is the desired Hermitian-Yang-Mills metric with the constant factor $c$. More generally, a Yang-Mills-Higgs pair $(A,\phi)$ (see the following definition) may produce a special  Yang-Mills-Higgs system. As a special case, this pair can be realized as the critical point of Yang-Mills-Higgs functional, called the strong Yang-Mills-Higgs pair (see the following definition). Our Yang-Mills-Higgs functional is a natural generation of Yang-Mills functional by replacing the curvature of Chern connection with that of Hitchin-Simpson connection. By calculate the first variation of functional, one can introduce the following flow equations for a pair $(A,\phi)$
\begin{eqnarray}
\begin{aligned}
  \frac{\partial A}{\partial t}&=(d^{\prime\prime}_A -d^\prime_A)(\mathcal{K}_{(A,\phi)}),\\
  \frac{\partial \phi}{\partial t}&=[\phi,\mathcal{K}_{(A,\phi)}].
\end{aligned}
\end{eqnarray}
It is easy to  check that the Hitchin pair will be preserved by this flow.

Now we introduce some concepts running through  this paper.
\begin{definition} Let $\{(E,\phi), h, d_A\}$ be a  Yang-Mills-Higgs system over an $n$-dimensional compact  K\"{a}hler manifold $(X,\omega)$. The associated Hitchin pair $(A,\phi)$ is called
\begin{itemize}
                                         \item a Yang-Mills-Higgs pair if it  satisfies  the following equation
\begin{align}\label{eq:p}
 d_A^\prime(\Lambda(F_A+[\phi,\phi^*]))=[\phi,\Lambda(F_A+[\phi,\phi^*])].
\end{align}
 \item a degenerate Yang-Mills-Higgs pair if   it is a  Yang-Mills-Higgs pair  with the property that $\det[\Lambda(F_A+[\phi,\phi^*])]=0$ at some point of $X$.
\item a strong Yang-Mills-Higgs pair if it is subject to the equations
\begin{eqnarray}\label{eq:o}
\begin{aligned}
  d_A^\prime(\Lambda(F_A+[\phi,\phi^*]))&=d_A^{\prime\prime}(\Lambda(F_A+[\phi,\phi^*]))=0,\\
  [\Lambda(F_A+[\phi,\phi^*]), \phi]&=[\Lambda(F_A+[\phi,\phi^*]),\phi^*]=0.
\end{aligned}
\end{eqnarray}
                                         \item a Hermitian-Yang-Mills-Higgs pair if   it satisfies the equation
                                          \begin{align}\label{po}
                                            i\Lambda(F_A+[\phi,\phi^*])=\lambda Id_E
                                          \end{align}
                                          for the  constant $\lambda=\frac{2\pi n}{\int_X\omega^n}\mu_\omega(E)$.
                                       \end{itemize}
\end{definition}
 Obviously a Hermitian-Yang-Mills-Higgs pair is a strong Yang-Mills-Higgs pair, and a  Hermitian-Yang-Mills-Higgs pair is a degenerate Yang-Mills-Higgs pair if and only if $\deg_\omega(E)=0$.  Existence of such pairs is a strong constraint on the Yang-Mills-Higgs system. In section 2, we will  exhibit the
mutual restriction of these constraints and stability conditions in algebraic geometry. For example, strongness and degeneracy conditions together generally   force the Higgs bundle to split, then by the principle of curvature decreases in Higgs subbundles, one can show that if the associated Hitchin pair is not a Hermitian-Yang-Mills-Higgs pair,  the Higgs bundle in the Yang-Mills-Higgs system cannot be semistable.

   Deformation of Yang-Mills-Higgs system is described by the deformation of the pair $(d^{\prime\prime}_A,\phi)$ which is controlled by a series of equations analogous to the Maurer-Cartan equation in the deformation theory of holomorphic vector bundle. The obstruction of deformation is also characterised by certain second order cohomology. Via calculating  the second variation of Yang-Mills-Higgs functional, one can fix which admitted deformation is stable for a given strong Yang-Mills-Higgs pair or if a strong Yang-Mills-Higgs pair is stable with respect to  a chosen  deformation. Due to the present of Higgs fields, there is no strong   Yang-Mills-Higgs pair that is stable  along arbitrary admitted deformation (e.g. $\mathbb{C^*}$-action on Higgs fields).  Such stability condition in the sense of differential geometry is reduced to judge the positive-definiteness of  a Hermitian quadratic form.

The remainder of this paper is organized as follows. We first introduce the Yang-Mills-Higgs functional for the higher dimensional case as a analog of that for Riemann surface. Next we  we study some properties of  the pairs associated with a Yang-Mills-Higgs system, for example,  the interaction with the stability of Higgs bundle and Higgs cohomology. In the last section, we consider the deformation of Hitchin pair and establish the notion of stability of strong Yang-Mills-Higgs pair along the admitted deformation via calculating the second variation of Yang-Mills-Higgs functional.

\section{Pairs Associated with the Yang-Mills-Higgs System }
\subsection{Yang-Mills-Higgs Functional}

Fix a Higgs bundle $(E,\phi)$ over a compact Riemann surface $\Sigma$, then the space $\mathcal{M}(E)$ of
holomorphic structures on $E$ as a smooth complex Hermitian vector
bundle can be identified with an affine space locally modeled on
$\Omega^{0,1}(\textrm{End}(E))$, i.e. $\mathcal{M}(E)=d^{\prime\prime}_A+\mathcal{A}^{0,1}(\textrm{End}(E))$
for a fixed operator
$d^{\prime\prime}_A$, and the Higgs field $\phi$ takes value in $H^0(\Sigma,K\otimes\textrm{End}(E))\simeq H^1(\Sigma,\textrm{End}(E))^\vee$ where $K$ denotes the canonical line bundle on $\Sigma$.
The
tangent space to $T^*\mathcal{M}(E)$ at any point can be naturally identified with the direct sum $\mathcal{A}^{0,1}(\textrm{End}(E))\oplus \mathcal{A}^{1,0}(\textrm{End}(E))$. Under  this identification the  metric on $T^*\mathcal{M}(E)$ is given by $$g((\psi_1,\phi_1),(\psi_2,\phi_2))=i(\int_\Sigma\psi_1^*\widehat{\wedge} \psi_2+\phi_1\widehat{\wedge}\phi_2^*+\int_\Sigma\psi_2^*\widehat{\wedge} \psi_1+\phi_2\widehat{\wedge}\phi_1^*)$$
where $(\psi_i,\phi_j)\in \mathcal{A}^{0,1}(\textrm{End}(E))\oplus \mathcal{A}^{1,0}(\textrm{End}(E))$ and for $\psi_1=f\otimes u, \psi_2=g\otimes v$ with $f,g\in\mathcal{A}^{0,1}(\Sigma),u,v\in \textrm{End}(E)$,  $\psi_1^*\widehat{\wedge} \psi_2=\sum_ih(v(e_i), u(e_i))\bar{f}\wedge g$.
Moreover there are compatible complex structures $I, J,K$ defined by
\begin{align*}
  I(\psi,\phi)&=(i\psi,i\phi),\\
  J(\psi,\phi)&=(i\phi^*, -i\psi^*),\\
  K(\psi,\phi)&=(-\phi^*, \psi^*).
\end{align*}
satisfying the usual quaternionic relations, namely $I^2=J^2=K^2=IJK=-1$. This defines the hyperk\"{a}hler structure on $T^*\mathcal{M}(E)$.
Let $\mathcal{G}$ denote the gauge group of $E$, which acts on $T^*\mathcal{M}(E)$ preserving the hyperk\"{a}hler structure.
The moment maps for this action are given by\cite{3}
\begin{align*}
 \mu_I&=F_A+[\phi,\phi^*],\\
 \mu_J&=-i(d^{\prime\prime}_A\phi+d^{\prime}_A\phi^*),\\
 \mu_K&=-d^{\prime\prime}_A\phi+d^{\prime}_A\phi^*.
\end{align*}

In analogy with the Yang-Mills functional, the full Yang-Mills-Higgs functional is defined to be the norm-square of the
hyperk\"{a}hler moment map, that is  we specify
\begin{equation}\label{eq:1}
YMH(A,\phi)=\int_\Sigma(|F_A+[\phi,\phi^*]|^2+4|d^{\prime\prime}_A\phi|^2)dV.
\end{equation}
We can generalize it to the higher dimensional  manifold $X$, thus we consider the following functional
\begin{equation}\label{eq:11}
YMH(A,\phi)=\int_X|F_A+[\phi,\phi^*]+d_A(\phi+\phi^*)|^2dV.
\end{equation}
\begin{proposition}Let  $(E,\phi)$ be  a polystable Higgs bundle on the K\"{a}hler manifold $(X,\omega)$,  then  we have the inequality
\begin{equation}\label{p}
||F_A||^2+3||[\phi,\phi^*]||^2+2||\nabla_A\phi||^2\geq2 \langle\mathbb{R}\lrcorner\phi,\phi\rangle,
  \end{equation}
  where the $L^2$-norm $||\cdot||$ and the global $L^2$-inner product $\langle\cdot,\cdot\rangle$ are with respect to the  Hermitian-K\"{a}hler metric
$g_\omega$  associated with K\"{a}hler form $\omega$ on $X$ and the Hermitian-Yang-Mills-Higgs metric $h$ on $E$, and $A$ is the Chern connection with respect to $h$,  $\mathbb{R}$ denotes the Ricci curvature tensor of $X$. \end{proposition}
\begin{proof}Generally we consider the following functional without the requirement that $\phi$ is a Higgs field
\begin{align*}
                  \mathcal{F}(A,\phi)&=||F_A+[\phi,\phi^*]||^2+4||d^{\prime\prime}_A\phi||^2\\
                  &=||F_A||^2+||[\phi,\phi^*]||^2+2\langle F_A,[\phi,\phi^*]\rangle.
\end{align*}
The term coupling the curvature and the Higgs field can be calculated as
\begin{align*}
  2\langle F_A,[\phi,\phi^*]\rangle=&2\int_X\sum_ig_\omega\otimes h(F_A(e_i),[\phi,\phi^*](e_i))=2\int_X\sum_i h((F_A)_{\mu\bar \nu}(e_i),[\phi\bar{^\mu},(\phi^*)^{\nu}](e_i))\\
  =&\int_X\sum_i h((F_A)_{\mu\bar \nu}\phi^{\bar\nu}(e_i),\phi\bar{^\mu}(e_i))-h(\phi^{\bar\nu}(F_A)_{\mu\bar \nu}(e_i),\phi\bar{^\mu}(e_i))\\
 &+\int_X h(\phi\bar{^\mu}(e_i),(F_A)_{\mu\bar \nu}\phi^{\bar\nu}(e_i))-h(\phi\bar{^\mu}(e_i),\phi^{\bar\nu}(F_A)_{\mu\bar \nu}(e_i))\\
  =&2\textrm{Re}\langle F_A\lrcorner\phi,\phi\rangle,
\end{align*}
where   $\phi^{\bar\nu}=g_\omega^{\mu\bar\nu}\phi_\mu$, $(\phi^*)^{\nu}=g_\omega^{\nu\bar\mu}\phi^*_{\bar\mu}$.
On the other hand, by the Bochner-Weitzenb\"{o}ck formula for bundle-valued 1-form\cite{6}
$$\{d_A,d^*_A\}\phi_\mu=\nabla^*_A\nabla_A\phi_\mu-\mathbb{R}_{\mu\bar\nu}\phi^{\bar\nu}-[(F_A)_{\mu\bar{\nu}},\phi^{\bar\nu}],$$
we obtain
$$ \langle F_A,[\phi,\phi^*]\rangle=||\nabla_A\phi||^2-||d_A\phi||^2-||d_A^*\phi||^2- \langle\mathbb{R}\lrcorner\phi,\phi\rangle.$$
Therefore we arrive at
\begin{equation*}
\mathcal{F}(A,\phi)=||F_A||^2+||[\phi,\phi^*]||^2+2||\nabla_A\phi||^2-2||d_A^\prime\phi||^2-2||(d_A^\prime)^*\phi||^2+2||d_A^{\prime\prime}\phi||^2-2 \langle\mathbb{R}\lrcorner\phi,\phi\rangle.
\end{equation*}
It follows from the K\"{a}hler identities $i[\Lambda,d_A^\prime]=-(d_A^{\prime\prime})^*$, $i[\Lambda,d_A^{\prime\prime}]=(d_A^{\prime})^*$
that
\begin{align*}
  ||d_A^\prime\phi||^2-||d_A^{\prime\prime}\phi||^2&=-i\langle\phi,[\Lambda,d_A^{\prime\prime}]d_A^\prime\phi+[\Lambda,d_A^\prime]d_A^{\prime\prime}\phi\rangle\\
  &=-i\langle\phi,[\Lambda F_A,\phi]-d_A^\prime[\Lambda,d_A^{\prime\prime}]\phi\rangle\\
  &=-i\langle\phi,[\Lambda F_A,\phi]\rangle-||(d_A^\prime)^*\phi||^2.
\end{align*}
Thus we find that the
functional  $\mathcal{F}(A,\phi)$ can be rewritten as
\begin{equation*}
  \mathcal{F}(A,\phi)=||F_A||^2+||[\phi,\phi^*]||^2+2||\nabla_A\phi||^2-2 \langle\mathbb{R}\lrcorner\phi+i[\Lambda F_A,\phi],\phi\rangle.
  \end{equation*}
For the Hermitian-Yang-Mills-Higgs metric, $i\Lambda F_A=\lambda Id-i\Lambda[\phi,\phi^*]$, thereby
\begin{align*}
i\langle[\Lambda F_A,\phi],\phi\rangle &=-i \langle[\Lambda[\phi,\phi^*],\phi],\phi\rangle =i \langle[\phi,\Lambda[\phi,\phi^*]],\phi\rangle  \\
 &=\int_X\sum_ih([\phi^{\bar\mu},\phi^*_{\bar\mu}](e_i),[\phi^*_{\bar\nu},\phi^{\bar\nu}](e_i))
  =-\int_X\sum_ih([[\phi^{\bar\mu},\phi^*_{\bar\mu}],\phi_{\nu}](e_i),\phi^{\bar\nu}(e_i))\\
  &=-\int_X\sum_ih([[\phi_{\nu},\phi^*_{\bar\mu}],\phi^{\bar\mu}](e_i),\phi^{\bar\nu}(e_i))
  =-\int_X\sum_ih([\phi_{\nu},\phi^*_{\bar\mu}](e_i),[\phi^{\bar\nu},(\phi^*)^{\mu}](e_i))\\
 &=-||[\phi,\phi^*]||^2.
\end{align*}
Then  by the semi-positivity of $\mathcal{F}(A,\phi)$, we deduce the inequality \eqref{p}.
\end{proof}
\subsection{ From Hitchin Pairs to Strong Yang-Mills-Higgs Pairs}

In this section, we study some properties of the pairs associated with a Yang-Mills-Higgs system.
\begin{proposition}  Let $\{(E,\phi), h, d_A\}$ be a  Yang-Mills-Higgs system over an $n$-dimensional compact  K\"{a}hler manifold $(X,\omega)$,
$(A,\phi)$ be the associated  Hitchin pair and  $0=E_0\subset E_1\subset\cdots\subset E_{l-1}\subset E_l=E$ be the  unique Harder-Narasimhan filtration of Higgs bundle $(E,\phi)$. Then the slope $\mu_\omega(E_1)$ of $E_1$ is not greater than the integral $\frac{\int_X\lambda(x)\omega^n}{2\pi n}$  for the largest eigenvalue $\lambda(x)$ of  the Hermitian matrix $\mathcal{K}_{(A,\phi)}|_{x}$ at $x\in X$.
\end{proposition}
\begin{proof}Let us  denotes by $E_{i}^\perp$ the orthogonal complement of $E_i$ in $E_{i+1}$, then the relation between the component $\mathcal{R}_{(A,\phi)}|_{E_1}$ of $\mathcal{R}_{(A,\phi)}$ restricted on $E_1$ and the  curvature $\mathcal{R}_{(A,\phi)}(E_1)$ with respect to the induced Hitchin-Simpson connection on $E_1$ is given by
 \begin{align*}
                                                                                                                                        \mathcal{R}_{(A,\phi)}|_{E_1}-\mathcal{R}_{(A,\phi)}(E_1)
 =&-\alpha_1\wedge\alpha_1^*-\alpha^{(1)}_2\wedge(\alpha^{(1)}_2)^*-\cdots-\alpha^{(l-2)}_{l-1}\wedge(\alpha^{(l-2)}_{l-1})^*\\&+\beta_1\wedge\beta_1^*+\beta^{(1)}_2\wedge(\beta^{(1)}_2)^*+\cdots+\beta^{(l-2)}_{l-1}\wedge(\beta^{(l-2)}_{l-1})^*,\end{align*}
 where $\alpha^{(i-1)}_i=\alpha_i\circ p_i\circ\cdots\circ p_2$ for $\alpha_i\in \Omega^{0,1}\otimes \textrm{Hom}(E_i^\perp,E_{i})$ coming from the decomposition of the holomorphic structure and the natural  projection $p_i$ from $E_i$ to $E_{i-1}$, $\beta^{(i-1)}_i=\beta_i\circ p_i\circ\cdots\circ p_2$ for $\beta_i\in \Omega^{1,0}\otimes \textrm{Hom}(E_i^\perp,E_{i})$ engendered by the decomposition of the Higgs field. Then \begin{align*}
                                                                      &\deg_\omega(E_1)=\int_Xc_1(E_1)\wedge\omega^{n-1}=\frac{i}{2\pi n}\int_X\textrm{Tr}(\Lambda\mathcal{R}_{(A,\phi)}(E_1))\omega^n\\
                                                                      =&\frac{1}{2\pi n}\int_X(\textrm{Tr}(i\Lambda\Theta|_{E_1})-|\alpha_1|^2-|\alpha^{(1)}_2|^2-\cdots-
                                                                      |\alpha^{(l-2)}_{l-1}|^2-|\beta_1|^2-|\beta^{(1)}_2|^2-\cdots-
                                                                      |\beta^{(l-2)}_{l-1}|^2)\omega^n\\
                                                                      \leq&\frac{1}{2\pi n}\textrm{rank}(E_1)\int_X\lambda\omega^n,
                                                                                         \end{align*}
where the property of   Hermitian matrix that any  diagonal element is  not bigger than the  largest eigenvalue plays a crucial role. Indeed, for any Hermitian matrix $H=[H_{ij}]=U\textrm{diag}\{\lambda_1,\cdots,\lambda_r\}U^*$ with  some unitary matrix $U$ and eigenvalues $\lambda_1,\cdots,\lambda_r$, one has $H_{ii}=\sum_j\lambda_jU_{ij}\bar U_{ij}\leq \lambda_{\max}\sum_jU_{ij}\bar U_{ij}=\lambda_{\max}(:=\max\{\lambda_1,\cdots,\lambda_r\})$.
\end{proof}

\begin{proposition}Let $(A,\phi)$ be a Yang-Mills-Higgs pair. The following facts are obvious:
\begin{enumerate}
  \item If $[\phi,\Lambda\Theta]$ is $\Delta_{d^\prime_A}$-harmonic, then $(A,\phi)$ is a strong Yang-Mills-Higgs pair.
  \item If $(A,\phi)$ is non-degenerate,  then  $\det(\Lambda\Theta)$ is a non-zero constant.
\end{enumerate}
\end{proposition}
\begin{proposition}Let $\{(E,\phi), h, d_A\}$  and $\{(E,\phi), \tilde h, d_{\tilde A}\}$  be two Yang-Mills-Higgs system with $h,\tilde{h}$ being   conformal to each other pointwise. If the corresponding Hitchin pairs $(A,\phi)$ and $(\tilde{A},\phi)$ are both Yang-Mills-Higgs pairs, then $\tilde{h}$ is just a rescaling of $h$.
\end{proposition}
\begin{proof}
We write $\tilde{h}=fh$ with the conformal  factor $f $ being a positive smooth function on $X$. Then
\begin{align*}
  i\Lambda\tilde{\Theta}&=\Lambda((d^\prime_A+(\tilde{h})^{-1}d^\prime_A\tilde{h}+d_A^{\prime\prime})^2+[\phi,\phi^{*_{\tilde{h}}}])\\
  &=i\Lambda(\Theta+\bar\partial(f^{-1}\partial f)Id)\\&=i\Lambda\Theta+(\frac{1}{2}\Delta\ln f)Id,
\end{align*}
which shows that the right hand sides of \eqref{eq:p} are conformal invariant.  Thereby the exact form $d\ln f$ is harmonic, so it has to vanish, thus $f$ is a constant.
\end{proof}

\begin{proposition}\label{q}A Hitchin  pair $(A,\phi)$ is a strong Yang-Mills-Higgs pair if and only if it satisfies the equation \begin{align}\label{eq:i}
 \mathcal{D}^*_{(A,\phi)}(F_A+[\phi,\phi^*]+d_A^\prime\phi+d_A^{\prime\prime}\phi^*)&=0.
\end{align}
\end{proposition}
\begin{proof}
By manipulations of the K\"{a}hler identities with Higgs field \eqref{eq:a}, we obtain
\begin{align*}
 &\mathcal{D}^*_{(A,\phi)}(F_A+[\phi,\phi^*]+d_A^\prime\phi+d_A^{\prime\prime}\phi^*)\\=&i([\Lambda,\mathcal{D}^{\prime\prime}_{(A,\phi)}]-[\Lambda,\mathcal{D}^{\prime}_{(A,\phi)}])(\Theta+d_A^\prime\phi+d_A^{\prime\prime}\phi^*)\\
=&i(\mathcal{D}^{\prime}_{(A,\phi)}-\mathcal{D}^{\prime\prime}_{(A,\phi)})(\Lambda\Theta)+i\Lambda(\mathcal{D}^{\prime\prime}_{(A,\phi)}-\mathcal{D}^{\prime}_{(A,\phi)})(\Theta+d_A^\prime\phi+d_A^{\prime\prime}\phi^*)
\end{align*}
From the Banchi identity for $A$, we have
\begin{align*}
  &(\mathcal{D}^{\prime\prime}_{(A,\phi)}-\mathcal{D}^{\prime}_{(A,\phi)})(\Theta+d_A^\prime\phi+d_A^{\prime\prime}\phi^*)\\
  =&[\phi-\phi^*, \Theta]+(d_A^{\prime\prime}-d_A^{\prime})[\phi,\phi^*]+[\phi-\phi^*,d_A^\prime\phi+d_A^{\prime\prime}\phi^*]+d_A^{\prime\prime}d_A^{\prime}\phi-d_A^{\prime}d_A^{\prime\prime}\phi^*\\
  =&[\phi,F_A]-[\phi^*,F_A]+ [F_A,\phi]-[F_A,\phi^*]\\
 =&0,
\end{align*}
where the second equality is due to $[\phi,d_A^{\prime}\phi]=-d_A^{\prime}(\phi\wedge\phi)=0$ and the Jacobi identity which implies $[\phi,[\phi,\phi^*]]=[\phi\wedge \phi,\phi^*]=0$. Therefore,
\begin{align*}
  &\mathcal{D}^*_{(A,\phi)}(F_A+[\phi,\phi^*]+d_A^\prime\phi+d_A^{\prime\prime}\phi^*)\\=&
  i(d_A^\prime-d_A^{\prime\prime})(\Lambda\Theta)-i[\phi,\Lambda\Theta]+i[\phi^*,\Lambda\Theta].
\end{align*}
 As a result,  the equations \eqref{eq:i} reduce to
\begin{eqnarray}\label{eq:n}
\begin{aligned}
  d^\prime_A(\Lambda\Theta)&=[\phi,\Lambda\Theta],\\
  d^{\prime\prime}_A(\Lambda\Theta)&=[\phi^*,\Lambda\Theta].
\end{aligned}\end{eqnarray}
We have to show that both sides of the equalities in \eqref{eq:n} vanish. Indeed, by K\"{a}hler identities we calculate
\begin{align*}
  ||d^{\prime\prime}_A(\Lambda\Theta)||^2&=\langle d^{\prime\prime}_A(\Lambda\Theta),[\phi^*,\Lambda\Theta]\rangle=\langle \Lambda\Theta, (d^{\prime\prime}_A)^*([\phi^*,\Lambda\Theta])\rangle\\
  &=\langle \Lambda\Theta, -i[\Lambda, d^{\prime}_A]([\phi^*,\Lambda\Theta])\rangle=i\langle \Lambda\Theta, \Lambda ([\phi^*, d^{\prime}_A\Lambda\Theta])\rangle\\
  &=i\langle \Lambda\Theta, \Lambda ([\phi^*, [\phi,\Lambda\Theta]])\rangle\\
  &=i\langle \Lambda\Theta, \Lambda ([\phi, [\Lambda\Theta,\phi^*]])\rangle-i\langle \Lambda\Theta, [\Lambda\Theta,\Lambda([\phi,\phi^*])]\rangle\\
  &=i\langle \Lambda\Theta, [\Lambda,d^{\prime\prime}_A]([\phi, \Lambda\Theta])\rangle\\
  &=-\langle d^{\prime}_A(\Lambda\Theta), [\phi, \Lambda\Theta]\rangle=-||[\phi, \Lambda\Theta]||^2=0,
\end{align*}
where the Jacobi identity has been applied to the forth equality, and the anti-Hermiticity of $\Theta$ yields $\langle \Lambda\Theta, [\Lambda\Theta,\Lambda([\phi,\phi^*])]\rangle=-\textrm{Re}\langle [\Lambda\Theta,\Lambda\Theta],\Lambda([\phi,\phi^*])\rangle=0$. This completes the proof.
\end{proof}

\begin{definition}A global  holomorphic section $s$ of the Higgs bundle $(E,\phi)$ is called $(\phi,\kappa)$-invariant if for the holomorphic 1-form $\kappa$ the equality $\phi(s) =\kappa\otimes s$ holds.
\end{definition}
\begin{proposition}Let  $(A,\phi)$ be  a  strong Yang-Mills-Higgs pair associated with a  Yang-Mills-Higgs system  $\{(E,\phi), h, d_A\}$.  If
   the Higgs field $\phi$ admits a non-trivial
$(\phi,\kappa)$-invariant section $s$,  then  we have the inequality
\begin{align*}
  ||\widetilde{d^\prime_A s}+\bar\kappa\tilde s||^2\geq||\kappa\tilde s||^2,
\end{align*}
where $\tilde s:=\Lambda\Theta(s)$. Obviously, when $s$ is covariant constant, the equality holds.
\end{proposition}
\begin{proof} It is known that the K\"{a}hler identities imply that
 \begin{align}\label{s}
 \Delta_{d^\prime_A}-\Delta_{d^{\prime\prime}_A}=i[\Lambda, F_A].
 \end{align}
 Therefore for any global holomorphic section $s $ we have
 \begin{align}\label{ss}
  ||d^\prime_As||^2=<i\Lambda F_A s,s>\geq0.
\end{align}
Moreover if $s$ is $(\phi,\kappa)$-invariant, then $<i\Lambda\Theta s,s>\geq0$ since $i\Lambda\Theta (s)=i\Lambda F_A (s)$.
 Now $(A,\phi)$ is   a strong Yang-Mills-Higgs pair
 then $\Lambda\Theta(s)$ is also a  $(\phi,\kappa)$-invariant section since  $d^{\prime\prime}_A(\Lambda\Theta)=[\phi,\Lambda\Theta]=0$.
 By the K\"{a}hler identities with Higgs field we have an analog of \eqref{s}
\begin{align}\label{sss}
\Delta_{\mathcal{D}^\prime_{(A,\phi)}}-\Delta_{\mathcal{D}^{\prime\prime}_{(A,\phi)}}=i[\Lambda, \mathcal{R}_{(A,\phi)}].
\end{align}
   Applying \eqref{ss} and \eqref{sss} to $\Lambda\Theta(s)$ yields
\begin{align*}
  ||\mathcal{D}^\prime_{(A,\phi)}\Lambda\Theta(s)||^2\geq ||\mathcal{D}^{\prime\prime}_{(A,\phi)}\Lambda\Theta(s)||^2.
\end{align*}
This is exactly the inequality in the proposition.
\end{proof}
\begin{proposition}\label{eq:y}If a  Yang-Mills-Higgs system  $\{(E,\phi), h, d_A\}$ over an $n$-dimensional   K\"{a}hler manifold $(X,\omega)$ admits   a  strong degenerate Yang-Mills-Higgs pair $(A,\phi)$, then $(A,\phi)$ must be  a Hermitian-Yang-Mills-Higgs pair if $(E,\phi)$ is  a semistable  Higgs bundle.
\end{proposition}
\begin{proof}Define
$F:=\Lambda\Theta(E)$ that is a proper holomorphic subbundle of $E$ since $\Lambda\Theta$ is degenerate and covariantly constant. The commutativity
of $\Lambda\Theta$ and $\phi$ guarantees that $F$ is a Higgs subbundle. Similarly $K:=\ker(\Lambda\Theta)$ is also a non-trivial Higgs subbundle of $E$. There is an orthogonal decomposition of $E$ as a $C^\infty$-bundle with respect to the Hermitian metric $h$: $E=F\oplus F^\bot$. However,  we have an isomorphism of $C^\infty$-bundles: $ F^\bot\simeq K$. In fact, if $u\in K$, then for any $v=\Lambda\Theta(s)\in F$ we have $h(u,v)=-h(\Lambda\Theta (u),s)=0$, i.e. $K\subset F^\bot$,  conversely,  $F^\bot\subset K$ is also due to the anti-Hermiticity of $\Theta$. So this $C^\infty$-decomposition is actually a holomorphic decomposition, which means that the second fundamental forms of  the subbundles $F$ and $K$ vanish. Hence $R_K=\pi_KR_E\pi_K$ where $R_E$ and $R_K$ are curvatures corresponding to connections $d_A$ on $E$ and $\pi_Kd_A\pi_K$ on $K$ respectively, and $\pi_K$ stands for the  projection to $K$. Therefore, we have
\begin{align*}
 \deg_\omega(K)&=\frac{i}{2\pi n}\int_X\textrm{Tr}(\Lambda R_K)\omega^{n} \nonumber\\
  &=-\frac{i}{2\pi n}\int_X\textrm{Tr}_K(\Lambda [\phi,\phi^*])\omega^{n}=-\frac{i}{2\pi }\int_X\textrm{Tr}_K([\phi,\phi^*])\wedge\omega^{n-1}=0,
 \end{align*}
which implies
  $\deg_\omega(F) =\deg_\omega(E)- \deg_\omega(K)= \deg_\omega(E)$,
 thus $\mu_\omega (F)\geq\mu_\omega(E)$. This  shows that $E=K$ if $(E,\phi)$ is a semistable Higgs bundle, then  $(A,\phi)$ is a Hermitian-Yang-Mills-Higgs pair.
 \end{proof}

 \begin{corollary}Suppose a  Yang-Mills-Higgs system  $\{(E,\phi), h, d_A\}$ over an $n$-dimensional   K\"{a}hler manifold $(X,\omega)$ admits   a  strong degenerate Yang-Mills-Higgs pair $(A,\phi)$.
 \begin{enumerate}
                                         \item Assume $E$ is of rank   2, if $deg_\omega(E)=0$, then $(A,\phi)$ is a Hermitian-Yang-Mills-Higgs pair, and if $deg_\omega(E)\neq0$,  the  Harder-Narasimhan filtration  associated with $(E,\phi)$ is  exactly $0\subset F\subset E$.
 \item  If $(E,\phi)$ is a semistable Higgs bundle over $X=\mathbb{P}^1$, then $E$ decomposes orthogonally into the direct sum of trivial Higgs line bundles.
 \item  Assume  the rank of $E$ is not less than 2, and  $X$ is an elliptic curve, then $(E,\phi)$ cannot be a stable Higgs bundle.
                                        \end{enumerate}
 \end{corollary}
 \begin{proof}
 (1)
  Suppose $(A,\phi)$ is not   a Hermitian-Yang-Mills-Higgs pair, then $F$ is a  line bundle such that $i\Lambda\Theta|_F=\kappa$. $\Lambda\Theta$ being covariantly yields that  $\kappa$ is  a non-zero constant which shows $\deg_\omega(E)=\deg_\omega(F)\neq0$.

 (2) Since $\deg(E)=0$, according to the classical Grothendieck theorem, we have $E\simeq\mathcal{O}(m_1)\oplus\cdots\oplus\mathcal{O}(m_r)$, where
the sum of  integers $m_1,\cdots,m_r$ that are unique up to the permutation is zero.  It follows from the   Higgs-semistability of $E$ that if not all $m_i(i=1,\cdots, r)$ vanish then for some positive integer $n$ there is    a non-trivial morphism $\tilde\phi:\mathcal{O}(k)\rightarrow\mathcal{O}(k^\prime -2)$ $(n^\prime\neq n)$ that factors through the Higgs morphism $\phi|_{\mathcal{O}(k)}:\mathcal{O}(k)\rightarrow E\otimes\mathcal{O}(-2)$ and the natural projection. Therefore there is a non-zero element belongs to
$H^0(\mathbb{P}^1,\mathcal{O}(k^\prime-k-2))$ which means that $k^\prime\geq k+2>k$. Hence, the semi-stability and zero degree together guarantee that there is no component as $\mathcal{O}(k), k\neq0$. As a result, $m_1=\cdots=m_r=0$, then the Higgs field on the trivial bundle $E\simeq \mathcal{O}^{\oplus r}$ should to be also trivial. Fix a component $L\simeq \mathcal{O}$ of $E$ so that $E\simeq L\oplus L^\perp$, then  $F_A|_L=\bar\partial(f^{-1}\partial f)-\alpha\wedge\alpha^*$ where  $f$ is a smooth function, and $\alpha\in \Omega^{0,1}\otimes \textrm{Hom}(L^\perp,  L)$ denotes the second  fundamental form. Hence $i\Lambda\Theta|_L=\frac{1}{2}\Delta\ln f+|\alpha|^2=0$. By Hopf's maximum principle, $f$ has to be a constant, thus $\alpha$ must vanish, in other words, $L^\perp \simeq\mathcal{O}^{\oplus r-1}$ as the holomorphic bundles. Then the induction on the rank $r$ of $E$ gives the conclusion.

   (3) Suppose $(E,\phi)$ is Higgs-stable, thus $(E,h, \phi)$ is a harmonic Higgs bundle. Firstly, we claim that $E$ has to split. Indeed, if not so, by Atiyah's results\cite{7}, $E$ is isomorphic to $E^\prime\otimes L$ where $E^\prime$ is an indecomposable  bundle of  rank $r$ and degree
zero with a global section,  and $L$ is  a  line
bundle  of degree zero, moreover, there is an exact sequence $0\rightarrow\mathcal{O}_\Sigma\rightarrow E^\prime\rightarrow E^{\prime\prime}\rightarrow0$ for an indecomposable  bundle $E^{\prime\prime}$ of  rank $r-1$ and degree zero, which implies $L$ is a proper subbundle of $E$. Since the canonical line bundle of $\Sigma$ is trivial, the Higgs field $\phi$ induces a morphism $\tilde\phi: L\rightarrow E$, or a    morphism $\tilde\phi:\mathcal{O}_X\rightarrow E^\prime$. By the Higgs-stability of $E$, the composition $p\circ \tilde\phi\in H^0(X,\mathcal{O}_X)$ of $\tilde\phi$ and the projection $p:E\rightarrow L$ cannot be an isomorphism,  which exhibits  a contradiction. Hence we deduce the claim.
Now from the decomposition  $E=E_1\oplus E_2$ one easily sees that $E$ is not stable. Assume $E$ is strictly semistable, then $\deg_\omega(E_1)=\deg_\omega(E_2)=0$. By  recursion, we find that $E$ can be decomposed into the direct sum of indecomposable holomorphic bundles of degree zero, i.e. $E=\oplus E_i$. Let $\tilde\phi_i=p_i\circ \phi_i$ be the composition of  $\phi_i=\phi|_{E_i}:E_i\rightarrow E$ and the projection $p_i:E\rightarrow E_i$, then $(E_i,\tilde\phi_i)$ are all Higgs bundles. By the same  arguments, for each $E_i$, if $\tilde\phi_i\neq0$ there is a line subbundle  $L_i\subseteq E_i$ of degree zero such that the restriction of $\tilde\phi_i$ on $L_i$ is an isomorphism, namely $L_i$ is a $\phi$-invariant proper subbundle of $E$,  which will contradict with the Higgs-stability of $E$.  The remaining case is $\tilde\phi_i=0$ for $\forall i$, then for each $E_i$ there is $E_j,j\neq i,$ such that $\Phi_{ij}=p_j\circ\phi_i:E_i\rightarrow E_j$ is a non-zero morphism. Write $E_i=L_i\otimes E_i^\prime$, and it follows from the fact that $(E_i^\prime)^\vee\simeq E_i^\prime$ and the multiplicative structure\cite{7} of $E_i^\prime$'s that $\Phi_{ij}|_{L_i}:L_i\xrightarrow{\sim}L_j$ which again contradicts with the Higgs-stability. So far, we only need to prove that if $E$ is not semistable, then $(E,\phi)$ is not Higgs semistable.  To show it, we consider the  Harder-Narasimhan filtration of $E$: $0=E_0\subset E_1\subset\cdots\subset E_{l-1}\subset E_l=E$, and let $E_i$ be the smallest subbundle among them containing $\phi(E_1)$. If $\phi(E_1)\neq0$,  we get a non-zero morphism $\overline{\phi}:E_1\rightarrow Gr_i(E)$ by the composition of $\phi$ and taking quotient. But It is known that $H^0(X,\textrm{Hom}(E_1,Gr_i(E)))=0$ if $i>1$ since $E_1$ and $Gr_i(E)$ are all semistable and $\mu_\omega(E_1)>\mu_\omega(Gr_i(E))$. Thus  $E_1$ is a Higgs subbundle with $\mu_\omega(E_1)>\mu_\omega(E)$. We complete the proof.
\end{proof}

\begin{definition} (\cite{j}) Let $(E,\phi)$ be a Higgs bundle over a K\"{a}hler manifold $X$.  We call the following complex of coherent $\mathcal{O}_X$-modules:
$$\mathcal{E}^\bullet=(E\xrightarrow{\phi}E\otimes \Omega^1_X\xrightarrow{\phi}E\otimes \Omega^2_X\rightarrow\cdots)$$ the Higgs complex, and define the Higgs cohomology $H^i(X,(E,\phi)):=\mathbb{H}^i(X,\mathcal{E}^\bullet)$ to be the hypercohomology of the Higgs complex.
\end{definition}
\begin{corollary}
Suppose  a  Yang-Mills-Higgs system  $\{(E,\phi), h, d_A\}$  admits   a  strong degenerate Yang-Mills-Higgs pair $(A,\phi)$.
  \begin{enumerate}
    \item If $(E,\phi)$ is a non-Higgs-semistable bundle of rank 2 and $X=\mathbb{P}^n (n\geq 2)$, then the Higgs cohomologies $H^i(\mathbb{P}^n,(E(m),\tilde\phi))$ vanish  for all $m\in \mathbb{Z}$, $i=1,\cdots, n-1$, where $(E(k):=E\otimes \mathcal{O}_{\mathbb{P}^n}(m), \tilde\phi:=\phi|_{E}\otimes Id|_{\mathcal{O}_{\mathbb{P}^n}(m)})$ is regarded as a Higgs bundle.
        \item  If $(E,\phi)$ is Higgs-stable  with $\rank(E)\geq2$, then the Higgs cohomology $H^0(X,(E,\phi))$ vanishes, and if $X$ is an $n$-dimensional Calabi-Yau manifold, the Higgs cohomology  $H^{2n}(X,(E,\phi))$ is also vanishes.
  \end{enumerate}

\end{corollary}
 \begin{proof}(1) Since $\Lambda\Theta\neq0$, $(E,\phi)$ splits into the direct sum of two Higgs line bundles, thus $(E,\phi)=(F,\phi|_F)\oplus(K,\phi|_K)$, where $F\simeq\mathcal{O}_{\mathbb{P}^n}(k)$ for $k=\deg(E)$, $K\simeq\mathcal{O}_{\mathbb{P}^n}$.  However, we note that  $\phi|_F,\phi|_K\in H^0(\mathbb{P}^n,\Omega^1_{\mathbb{P}^n})$ should be zero because $H^0(\mathbb{P}^n,\Omega^1_{\mathbb{P}^n})=H^{1,0}(\mathbb{P}^n,\mathbb{C})=\overline{H^{0,1}(\mathbb{P}^n,\mathbb{C})}=\overline{H^1(X,\mathcal{O}_{\mathbb{P}^n})}=0$ when $n\geq 2$, thereby $H^i(\mathbb{P}^n,(E(k),\tilde\phi))=H^i(\mathbb{P}^n,\mathcal{O}_{\mathbb{P}^n}(k+m))\oplus H^i(\mathbb{P}^n,\mathcal{O}_{\mathbb{P}^n}(m))=0$ for $i=1,\cdots,n-1$.

 (2) The complex of sheaves of
$C^\infty$-sections
$\mathcal{A}^0(E)\xrightarrow{\mathcal{D}^{\prime\prime}_{(A,\phi)}}\mathcal{A}^1(E)\xrightarrow{\mathcal{D}^{\prime\prime}_{(A,\phi)}}\mathcal{A}^2(E)\longrightarrow\cdots$
gives a fine resolution of the Higgs complex. Therefore
the hypercohomology of the Higgs complex
is isomorphic to the cohomology of the complex of global sections
$\Gamma(\mathcal{A}^0(E))\xrightarrow{\mathcal{D}^{\prime\prime}_{(A,\phi)}}\Gamma(\mathcal{A}^1(E))\xrightarrow{\mathcal{D}^{\prime\prime}_{(A,\phi)}}\Gamma(\mathcal{A}^2(E))\longrightarrow\cdots$.
Assume that there is a non-trivial section $s\in \Gamma(\mathcal{A}^0(E))$ satisfies $d^{\prime\prime}_As=\phi(s)=0$, thus $\phi$ can be viewed as a $(\phi,0)$-invariant section. Then since $\Lambda\Theta=0$, we have $d^\prime_As=0$, which means $s$ may generate a flat line bundle $L\subset E$ with the trivial  Higgs field.
Hence $\deg(L)=\deg(E)=0$ that will contradict with the Higgs-stability of $E$. For a Calabi-Yau manifold $X$, $\Omega^n_X\simeq\mathcal{O}_X$, then by Serre duality, $\mathbb{H}^{2n}(X,\mathcal{E}^\bullet)\simeq(\mathbb{H}^0(X,(\mathcal{E}^\bullet)^\vee))^\vee$. The  Hitchin pair  on the dual stable Higgs bundle  $(E^\vee, h^\vee, \phi^{\vee})$ is also a strong degenerate Yang-Mills-Higgs pair, so the previous conclusion implies the vanishing of $\mathbb{H}^{2n}(X,\mathcal{E}^\bullet)$.
\end{proof}

Let   $\{(E, \phi),h,d_A)$ be the Yang-Mills-Higgs system  as that in the Proposition \ref{eq:y}. If the limit $\lim\limits_{c\to 0}(E,c\phi)$ exists (for example, when $X$ is a smooth projective manifold), it must be  a fixed point of the $\mathbb{C}^*$-action which implies that it carries the structure of   a system of Hodge bundles. More precisely, this means that $E$ with respect to the
limiting holomorphic structure splits holomorphically as a sum $E=\oplus_{i=1}^l E_i$ and  that the limiting Higgs field  is given by a collection of holomorphic
maps $\phi_i:E_i\rightarrow E_{i+1}\otimes \Omega^1_X, 1 \leq i \leq k$ (with the convention that $E_{l+1}=0$). Moreover each subbundle $E_i$ splits as $E_i=F_i\oplus K_i$ with $F_i\subset F, K_i\subset K$ and $\phi_i$ decomposes as $\phi_i=\theta_i\oplus\vartheta_i$ with $\theta_i: F_i\rightarrow F_{i+1}\otimes\Omega^1_X, \vartheta_i: K_i\rightarrow K_{i+1}\otimes\Omega^1_X$, thus $(\oplus_{i=1}^l F_i,\oplus_{i=1}^l\theta_i)$ and $(\oplus_{i=1}^l K_i,\oplus_{i=1}^l\vartheta_i)$ are both systems of Hodge bundles. Since the deformation changes
the holomorphic structures of $E, F$ and $ K$, but not their isomorphism classes as differentiable
complex vector bundles, hence their degrees remain unchanged, namely $\deg_\omega(\oplus F_i)=\deg_\omega(E)$ and $\deg_\omega(\oplus K_i)=0$.
\begin{corollary}
Suppose $X$ is a smooth projective manifold. If $c_2(E)=c_2(F)$,  the holomorphic tangent bundle $T_X$ is a semistable bundle and the map $\vartheta_{l-1}$ is a non-zero injective map,  then $\mu_\omega(K_{l-1})\leq\frac{n}{n+1}\mu_\omega(\Omega^1_X)$. In particular, if $X$ is a Riemann surface and $\vartheta_{l-1}$ satisfies the same assumption as above, we have $\mu_\omega(K_{l-1})\leq\frac{1}{2}\mu_\omega(\Omega^1_X)$.
\end{corollary}
\begin{proof}Since $\Lambda\Theta|_K=0$, thus $(K,\phi|_K)$ is a polystable Higgs bundle with trivial first Chern class, $c_2(K)=c_2(E)-c_2(F)=0$, then  $K$ is a harmonic Higgs bundle, thereby the corresponding  limiting systems of Hodge bundles   is a variation of Hodge structures which implies that $(\oplus_{i=1}^l K_i,\oplus_{i=1}^l\vartheta_i)$ is a semistable Higgs bundle.
  Let $P$ be the the first term in the Harder-Narasimhan filtration of $K_{l-1}$, which is semistable since $P$ has maximal slope among the
subbundles of $K_{l-1}$. Consider the injective map $\vartheta_{l-1}:K_{l-1}\otimes T^1_X\rightarrow K_{l}$, then $P\oplus \vartheta_{l-1}(P\otimes T_X)$  is a Higgs subbundle due to the vanishing of  $\vartheta_l$.
Therefore $\deg_\omega(P)+\deg_\omega(P\otimes T_X)\leq0$. Since $P\otimes T_X$ is a semistable bundle, then we have
$$\mu_\omega(P\otimes T_X)=\mu_\omega(P)+\mu_\omega(T_X)\leq -\frac{\deg_\omega(P)}{\rank(P\otimes T_X)},$$
thus
 $$(\rank(P\otimes T_X)+\rank(P))\mu_\omega(P)\leq\rank(P\otimes T_X)\mu_\omega(\Omega^1_X).$$
 Hence
 $$\mu_\omega(K_{l-1})\leq\mu_\omega(P)\leq\frac{\rank(P\otimes T_X)}{\rank(P\otimes T_X)+\rk(P)}\mu_\omega(\Omega^1_X)=\frac{n}{n+1}\mu_\omega(\Omega^1_X).$$
\end{proof}

\section{Stable Yang-Mills-Higgs Pairs}
\subsection{Deformation of the Hitchin Pair }
Let $(A_t=A_0+t\alpha, \phi_t=\phi_0+t\beta)$ be a family of the Hitchin pairs  with one parameter $t\in\mathbb{C}$, where $(A_0,\phi_0)$ is a fixed Hitchin pair, $\alpha\in\mathcal{A}^1_X(\textrm{End}(E))$, $\beta\in \mathcal{A}^{1,0}_X(\textrm{End}(E))$.
Then the deformation pair $(\alpha^{0,1},\beta)$ is subject to the following equations
\begin{eqnarray}\label{eq:l}
\begin{aligned}
  td^{\prime\prime}_{A_0}\alpha^{0,1}+t^2\alpha^{0,1}\wedge\alpha^{0,1}=0,\\
td^{\prime\prime}_{A_0}\beta+t[\alpha^{0,1},\phi_0]+t^2[\alpha^{0,1},\beta]=0,\\
t[\phi_0,\beta]+t^2\beta\wedge\beta=0.
\end{aligned}
\end{eqnarray}
 \begin{definition}\begin{enumerate}
                     \item If $d^{\prime\prime}_{A_0}\alpha^{1,0}=d^{\prime\prime}_{A_0}\beta=0$, then $(\alpha^{0,1},\beta)$ is called the holomorphic deformation pair.
 \item If one  expresses $\alpha=\alpha_0+\Sigma_{i\geq 1}\alpha_it^i$  and $\beta=\beta_0+\Sigma_{i\geq 1}\beta_it^i$, then $(a_0^{0,1},\beta_0)$ is called the infinitesimal deformation pair.
  \end{enumerate}
  \end{definition}
 The infinitesimal deformation pair $(a_0^{0,1},\beta_0)$ satisfies $\mathcal{D}^{\prime\prime}_{(A_0,\phi_0)}(a_0^{0,1}+\beta_0)=0$, thus $[a_0^{0,1}+\beta_0]\in H^1(X,(\textrm{End}(E),\tilde\phi_0))$ where $(\textrm{End}(E),\tilde\phi_0)$ is viewed as a Higgs bundle via  the induced  Higgs field $\tilde\phi_0=\phi_0|_E\otimes Id|_{E^\vee}+Id|_{E}\otimes\phi_0^\vee|_{E^\vee}$.
Let $\pi:\Omega^1_X\rightarrow X$ be the holomorphic cotangent bundle on $X$, then we have $E\simeq\pi_*\mathcal{S}$ and $(E\xrightarrow{\phi_0} E\otimes\Omega^1_X)\simeq\pi_*(\mathcal{S}\otimes(\mathcal{O}_{\Omega^1_X}\xrightarrow{\Phi}\pi^*\Omega^1_X))$ for a locally free $\mathcal{O}_{\Omega^1_X}$-sheaf $\mathcal{S}$ where $\Phi\in H^0(\pi^*\Omega^1_X)$ is
the tautological section\cite{4}, which induces a Koszul complex\cite{8} $$K_\bullet(\Phi)=(0\rightarrow\wedge^n\pi^*T_X\rightarrow\wedge^{n-1}\pi^*T_X\rightarrow\cdots\rightarrow\pi^*T_X\xrightarrow{\Phi^\vee}\mathcal{O}_{\Omega^1_X}\rightarrow0).$$
The zero scheme of $\Phi$ can be identified with $X$, hence $K_\bullet(\Phi)$ is a projective resolution of $\mathcal{O}_X$.
By definition, $\ext^i(\mathcal{O}_X\otimes_{\mathcal{O}_{\Omega^1_X}}\mathcal{S},\mathcal{S})=\mathbb{H}^i(\mathcal{H}om(K_\bullet(\Phi)\otimes_{\mathcal{O}_{\Omega^1_X}}\mathcal{S},\mathcal{S}))=
\mathbb{H}^i(K^\vee_\bullet(\Phi)\otimes_{\mathcal{O}_{\Omega^1_X}}\mathcal{E}nd(\mathcal{S}))=\mathbb{H}^i(\pi_*(K^\vee_\bullet(\Phi)\otimes_{\mathcal{O}_{\Omega^1_X}}\mathcal{E}nd(\mathcal{S})))$, therefore $[a_0^{0,1}+\beta_0]\in\ext^1(\mathcal{O}_X\otimes_{\mathcal{O}_{\Omega^1_X}}\mathcal{S},\mathcal{S})=\ext^1(\mathcal{O}_X,\mathcal{E}nd_{\mathcal{O}_{\Omega^1_X}}(\mathcal{S}))$.
\begin{proposition}If $\ext^2(\mathcal{O}_X,\mathcal{E}nd_{\mathcal{O}_{\Omega^1_X}}(\mathcal{S}))=0$, the solution of equations \eqref{eq:l} exists.
\end{proposition}
\begin{proof}
From \eqref{eq:l} the equations that the higher order terms of $\alpha,\beta$ should obey read
\begin{eqnarray}\label{eq:k}
\mathcal{D}^{\prime\prime}_{(A_0,\phi_0)}(a_k^{0,1}+\beta_k)+\sum_{i+j=k-1}(a_i^{0,1}+\beta_j)\wedge(a_j^{0,1}+\beta_j)=0
\end{eqnarray}
for $\forall k\geq1$. Let $\mathcal{H}^k$ denote the space of harmonic $k$-forms valued in $\textrm{End}(E)$ corresponding to the Laplacian $\Delta_{\mathcal{D}^{\prime\prime}_{(A_0,\phi_0)}}$, then there are isomorphisms $H^k(X,\textrm{End}(E))\simeq\mathbb{H}^k(X,\mathcal{E}\emph{nd}_{\mathcal{O}_X}(\mathcal{E}))\simeq\mathcal{H}^k$, and we have the operator $Q^{(k)}:
L^2(\mathcal{A}^k_X(\textrm{End}(E)))\rightarrow L^2(\mathcal{A}^k_X(\textrm{End}(E)))$ such that $\Delta_{\mathcal{D}^{\prime\prime}_{(A_0,\phi_0)}}\circ Q^{(k)}=Id-p_{\mathcal{H}^k}$ where $p_{\mathcal{H}^k}$ stands for the projection on the space $\mathcal{H}^k$. One can easily check that $Q^{(k)}$ commutes with $\mathcal{D}^{\prime\prime}_{(A_0,\phi_0)}$ and $(\mathcal{D}^{\prime\prime}_{(A_0,\phi_0)})^*$, then define the operator $G^{(k)}=(\mathcal{D}^{\prime\prime}_{(A_0,\phi_0)})^*\circ Q^{(k)}=Q^{(k)}\circ(\mathcal{D}^{\prime\prime}_{(A_0,\phi_0)})^*:L^2(\mathcal{A}^k_X(\textrm{End}(E)))\rightarrow L^2(\mathcal{A}^{k-1}_X(\textrm{End}(E)))$.

By assumption, $\Delta_{\mathcal{D}^{\prime\prime}_{(A_0,\phi_0)}}\circ Q^{(2)}=Id$, hence $\{\mathcal{D}^{\prime\prime}_{(A_0,\phi_0)},G^{(2)}\}=Id$. We put
\begin{align*}
 \alpha^{0,1}_k&=-p^{0,1}(\sum_{i+j=k-1}G^{(2)}((a_i^{0,1}+\beta_j)\wedge(a_j^{0,1}+\beta_j))) \\
  \beta_k&=-p^{1,0}(\sum_{i+j=k-1}G^{(2)}((a_i^{0,1}+\beta_j)\wedge(a_j^{0,1}+\beta_j))),
\end{align*}
where $p$ denotes the projection on the corresponding space, which provide the desired solution of \eqref{eq:k} via the induction on $k$.

 In order to complete the proof, we have to show the formal power series  $\alpha^{0,1}=\alpha^{0,1}_0+\Sigma_{k\geq 1}\alpha^{0,1}_kt^k$  and $\beta=\beta_0+\Sigma_{k\geq 1}\beta_kt^k$ are convergent,  and furthermore,  are  smooth sections. The method applied here is standard due to Kodaira-Spencer\cite{9}. We denote by $||\xi||_s$ the Sobolev norm of the section $\xi\in\mathcal{A}^k_X(\textrm{End}(E))$ which is given by the
sum of the $L^2$-norms of $i$-th derivative of $\xi$ for all $i \leq s$, where $s$ is a sufficiently large  integer
compared to $ 2\dim_\mathbb{C}X $. It follows from  the standard estimate of elliptic differential operators that
\begin{align*}
  ||G^{(2)}((a_i^{0,1}+\beta_j)\wedge(a_j^{0,1}+\beta_j))||_s&<C_1||(a_i^{0,1}+\beta_j)\wedge(a_j^{0,1}+\beta_j))||_{s-1}\\
  &<C_2||a_i^{0,1}+\beta_j||_s||a_i^{0,1}+\beta_j||_s
\end{align*}
with  positive constants $C_1$, $C_2$ depend only on $s$ and the manifold $X$. Then by induction on $k$  there exists a constant $C_3$ such that
$$ ||a_k^{0,1}+\beta_k||_s<[\frac{k+1}{2}]C_3^k||a_0^{0,1}+\beta_0||_s^{k+1}.$$
Therefore, if we choose suitably $a_0^{0,1}+\beta_0\in\mathcal{H}^1$ such that $|t|||a_0^{0,1}+\beta_0||_s$ is  sufficiently small, then we may deduce the convergence. By Sobolev's fundamental lemma, $a^{0,1}+\beta\in C^m(\mathcal{A}_X^1(\textrm{End}(E)))$ for $m=s-1-\dim_\mathbb{C}X$.  On the other hand, we note that $\Delta_{\mathcal{D}^{\prime\prime}_{(A_0,\phi_0)}}(a^{0,1}+\beta)+(\mathcal{D}^{\prime\prime}_{(A_0,\phi_0)})^*((a^{0,1}+\beta)\wedge(a^{0,1}+\beta))=0$ which is an elliptic PDE for sufficiently small $a^{0,1}+\beta$, hence $a^{0,1}+\beta\in C^\infty(\mathcal{A}_X^1(\textrm{End}(E)))$.
\end{proof}
\subsection{Stable Yang-Mills-Higgs Pairs}
Let us consider the second variation of
the Yang-Mills-Higgs functional, thus we calculate
$$\frac{d^2}{dt^2}|_{t=0}YMH(A_t,\phi_t)=\textrm{Re}\langle\Upsilon,\mathcal{D}^*_{(A_0,\phi_0)}\mathcal{D}_{(A_0,\phi_0)}\Upsilon+\Upsilon^*\lrcorner\mathcal{R}_{(A_0,\phi_0)}\rangle,$$
where $\Upsilon=-(\alpha^{0,1})^*+\alpha^{0,1}+\beta+\beta^*\in\mathcal{A}^1_X(\textrm{End}(E))$ associated with the  deformation pair $(\alpha^{0,1},\beta)$.
\begin{definition}A strong Yang-Mills-Higgs pair $(A_0,\phi_0)$ on a holomorphic vector bundle $(E,h)$ is called the  semi-stable (stable, unstable) Yang-Mills-Higgs pair along the given deformation pair $(\alpha^{0,1},\beta)$ if the following condition is satisfied
\begin{eqnarray}\label{eq:m}
\textrm{Re}\langle\Upsilon,\mathcal{D}^*_{(A_0,\phi_0)}\mathcal{D}_{(A_0,\phi_0)}\Upsilon+\Upsilon^*\lrcorner\mathcal{R}_{(A_0,\phi_0)}\rangle\geq(>,<)0,
\end{eqnarray}
and is called the weakly semi-stable Yang-Mills-Higgs pair if for the arbitrary admitted holomorphic deformation pair $(a^{0,1}\neq0, \beta)$ the inequality \eqref{eq:m} holds.
\end{definition}
\begin{proposition}\label{prpo:sd}If $(A_0,\phi_0)$ is a stable Yang-Mills-Higgs pair along the  deformation pair $(\alpha^{0,1},\beta)$, then we have the inequality
\begin{align}\label{l}
 Q_{(A_0,\phi_0)}(\alpha,\tilde\beta):= i\langle\alpha\circ\alpha+\tilde\beta\circ\tilde\beta,\Lambda\Theta_{(A_0,\phi_0)}\rangle+||d_{A_0}^*\alpha||^2-||\tilde\phi_0\circ\tilde\beta||^2>0,
\end{align}
where $Q_{(A_0,\phi_0)}$ is a hermitian quadratic form on the space of deformation pairs, $\alpha=-(\alpha^{0,1})^*+\alpha^{0,1}$, $\tilde\beta=\beta+\beta^*$, $\tilde\phi_0=\phi_0+\phi_0^*$, and the action $\circ$   is defined by $$\Omega\circ\Xi=\Omega\lrcorner\Xi^{1,0}+\Xi^{0,1}\lrcorner\Omega$$ for any $\Omega\in\mathcal{A}^k(\End(E))$ and $\Xi\in\mathcal{A}^1(\End(E))$. In particular, if  $\alpha$ is parallel with respect to the Hitchin-Simpson connection associated with $(A_0,\phi_0)$,  the inequality \eqref{l} is rewritten as
\begin{align}\label{ll}
 2i\langle\alpha,\Lambda\Theta_{(A_0,\phi_0)}\star\alpha\rangle -i\langle\alpha,\Lambda(\Theta_{(A_0,\phi_0)}\star\alpha)\rangle
 +||\nabla_{A_0}\alpha||^2-\langle\alpha,\mathbb{R}\circ\alpha\rangle&\nonumber\\
 -i\langle\tilde\beta,\Lambda\Theta_{(A_0,\phi_0)}\star\tilde\beta\rangle
  -||\tilde\phi_0\circ\tilde\beta||^2&>0,
\end{align}
where the action $\star$   is defined by
$$\Omega\star\Xi=[\Omega,\Xi^{1,0}]+[\Xi^{0,1},\Omega].$$
\end{proposition}
\begin{proof}From the proof of Proposition \ref{q}, we have seen that
 \begin{align*}
 \frac{d}{dt}|_{t=0}YMH(A_t,\phi_t)=&2i\langle(\alpha^{0,1})^*,d^\prime_{A_0}\Lambda\Theta_{(A_0,\phi_0)}\rangle+2i\langle\beta,[\phi_0,\Lambda\Theta_{(A_0,\phi_0)}]\rangle\\
 &+2i\langle\alpha^{0,1},d^{\prime\prime}_{A_0}\Lambda\Theta_{(A_0,\phi_0)}\rangle-2i\langle\beta^*,[\phi_0^*,\Lambda\Theta_{(A_0,\phi_0)}]\rangle.
 \end{align*}
 Therefore,  the second variation is given by
\begin{align*}
  &\frac{d^2}{dt^2}|_{t=0}YMH(A_t,\phi_t)\\
  =&-2i\langle(\alpha^{0,1})^*,[(\alpha^{0,1})^*,\Lambda\Theta_{(A_0,\phi_0)}]-id^\prime_{A_0}(d^\prime_{A_0})^*(\alpha^{0,1})^*-id^\prime_{A_0}(d^{\prime\prime}_{A_0})^*\alpha^{0,1}\rangle\\
  &+2i\langle(\alpha^{0,1})^*,d^\prime_{A_0}\Lambda([\beta,\phi_0^*]+[\phi_0,\beta^*])\rangle+2i\langle\beta,[\beta,\Lambda\Theta_{(A_0,\phi_0)}]\rangle\\
  &+2\langle\beta,[\phi_0,(d^\prime_{A_0})^*(\alpha^{0,1})^*+(d^{\prime\prime}_{A_0})^*\alpha^{0,1}]\rangle
   +2i\langle\beta,[\phi_0,\Lambda([\beta,\phi_0^*]+[\phi_0,\beta^*])]\rangle\\
   &+2i\langle\alpha^{0,1},[\alpha^{0,1},\Lambda\Theta_{(A_0,\phi_0)}]+id^{\prime\prime}_{A_0}(d^\prime_{A_0})^*(\alpha^{0,1})^*+id^{\prime\prime}_{A_0}(d^{\prime\prime}_{A_0})^*\alpha^{0,1}\rangle\\
 &+2i\langle\alpha^{0,1},d^{\prime\prime}_{A_0}\Lambda([\beta,\phi_0^*]+[\phi_0,\beta^*])\rangle-2i\langle\beta^*,[\beta^*,\Lambda\Theta_{(A_0,\phi_0)}]\rangle\\
&-2\langle\beta^*,[\phi_0^*,(d^\prime_{A_0})^*(\alpha^{0,1})^*+(d^{\prime\prime}_{A_0})^*\alpha^{0,1}]\rangle
-2i\langle\beta^*,[\phi_0^*,\Lambda([\beta,\phi_0^*]+[\phi_0,\beta^*])]\rangle\\
=&-4i\textrm{Re}\langle\alpha^{0,1}\lrcorner(\alpha^{0,1})^*+\beta\lrcorner\beta^*,\Lambda\Theta_{(A_0,\phi_0)}\rangle+4||(d^\prime_{A_0})^*(\alpha^{0,1})^*||^2+4\textrm{Re}\langle(d^\prime_{A_0})^*(\alpha^{0,1})^*,(d^{\prime\prime}_{A_0})^*\alpha^{0,1}
\rangle\\&-4||\phi_0\lrcorner\beta^*||^2-4\textrm{Re}\langle\phi_0\lrcorner\beta^*,\beta\lrcorner\phi_0^*\rangle\\
=&-4i\langle\alpha^{0,1}\lrcorner(\alpha^{0,1})^*+\beta\lrcorner\beta^*,\Lambda\Theta_{(A_0,\phi_0)}\rangle+2||d_{A_0}^*\alpha||^2-2||\phi_0\lrcorner\beta^*+\beta\lrcorner\phi_0^*||^2,
\end{align*}
which exhibits the Proposition the inequality \eqref{l}. If $\mathcal{D}_{(A_0,\phi_0)}\alpha=0$, then $\alpha$ is $d_{A_0}$-closed and $\tilde\phi_0$-invariant. By means of the Bochner-Weitzenb\"{o}ck
formula which leads to
\begin{align*}
 ||d_{A_0}^*\alpha||^2&=||\nabla_{A_0}\alpha||^2-\langle\alpha,\mathbb{R}\circ\alpha+\Theta_{(A_0,\phi_0)}\circ\alpha\rangle\\
 &=||\nabla_{A_0}\alpha||^2-\langle\alpha,\mathbb{R}\circ\alpha\rangle-i\langle\alpha,\Lambda(\Theta_{(A_0,\phi_0)}\star\alpha)-\Lambda\Theta_{(A_0,\phi_0)}\star\alpha\rangle,
\end{align*}
we immediately obtain the inequality the inequality \eqref{ll}.
\end{proof}

\begin{example}Assume that $v$ is a non-zero (1,0)-type  vector field that is  parallel with respect to the connection determined by the K\"{a}hler metric on $X$, and $ \Pi$ is a non-zero $\Delta_{d_{A_0}}$-harmonic  (1,1)-form valued in $\End(E)$, i.e. $d_{A_0}\Pi=d^*_{A_0}\Pi=0$, and they together satisfy
$$(\nabla_{A_0})_v\Pi=0, [\phi_0,v\lrcorner\Pi]=0, [v\lrcorner\Pi, v\lrcorner\Pi]=0,$$
 then  $v\lrcorner\Pi$ is also $\Delta_{d_{A_0}}$-harmonic, since $d_{A_0}(v\lrcorner\Pi)=(\nabla_{A_0})_v\Pi-v\lrcorner d_{A_0}\Pi=0$, and $\langle d_{A_0}^*(v\lrcorner\Pi),\theta\rangle=\langle\Pi,\bar v^\vee\wedge d_{A_0}\theta\rangle=-\langle d_{A_0}^*\Pi,\bar v^\vee\wedge \theta\rangle=0$ for $\forall\theta\in C^\infty(\End(E))$. Hence $(\alpha^{0,1}, \beta)=(v\lrcorner\Pi,\phi_0)$ gives rise to a holomorphic deformation. Moreover if the chosen pair $(A_0,\phi_0)$ is Hermitian and degenerate, then $Q(\alpha,\beta)=-4||[\phi_0,\phi_0^*]||^2$, thus $(A_0,\phi_0)$ is not stable along such deformation direction.
\end{example}

\begin{corollary}\begin{enumerate}
                     \item A Hermitian-Yang-Mills-Higgs pair $(A_0,\phi_0)$ is stable along the  deformation pair $(\alpha^{0,1},\beta)$ if and only if we have
\begin{align}\label{eq:n}
||d_{A_0}^*\alpha||^2> ||\phi_0\lrcorner\beta^*+\beta\lrcorner\phi_0^*||^2.
\end{align}
\item A stable Yang-Mills-Higgs pair $(A_0,\phi_0)$ along the  deformation pair $(\alpha^{0,1},\beta)$ on a Riemann surface  satisfies
\begin{align}\label{y}
 &\langle\tilde\beta,F_{A_0}\circ\tilde\beta-[\phi_0,\tilde\beta]\lrcorner\phi_0-\phi_0^*\lrcorner[\tilde\beta,\phi_0^*]\rangle+\langle\alpha,d_{A_0}d_{A_0}^*\alpha-\Theta_{(A_0,\phi_0)}
  \circ\alpha\rangle>0.
\end{align}
\item If there exists a weakly semi-stable Yang-Mills-Higgs pair $(A_0,\phi_0)$ with the property that $\Lambda\Theta_{(A_0,\phi_0)}$ and $\Lambda[\phi_0,\phi_0^*]$ are both non-zero on a compact K\"{a}hler manifold $X$, then the singular homology $H_1(X,\mathbb{R})$ vanishes.
     \end{enumerate}
 \end{corollary}
 \begin{proof}
(1) If $(A_0,\phi_0)$ is the Hermitian-Yang-Mills-Higgs pair  the stability condition \eqref{eq:m} obviously reduces to the inequality \eqref{eq:n}.

 (2) For  the case of Riemann surafce, we should note that the operator $\Lambda$ on 2-forms is an isometry with respect to the K\"{a}hler matric, hence
\begin{align*}
 4i\langle\beta\lrcorner\beta^*,\Lambda\Theta_{(A_0,\phi_0)}\rangle&=-2\langle\beta+\beta^*,\Theta_{(A_0,\phi_0)}\lrcorner\beta+\beta^*\lrcorner\Theta_{(A_0,\phi_0)}\rangle\\
 &=-2\langle\beta+\beta^*,F_{A_0}\circ(\beta+\beta^*)\rangle-4\langle[\beta,\beta^*],[\phi_0,\phi_0^*]\rangle\\
 &=-2\langle\beta+\beta^*,F_{A_0}\circ(\beta+\beta^*)\rangle-4||[\beta,\phi_0^*]||^2,
\end{align*}
where the third equality is due to the Jacobi identity. Then one can easily check the inequality \eqref{y}.

(3) By $\partial\bar\partial$-lemma, the space $V$ consisting of the closed (0,1)-forms on $X$  is isomorphic to the cohomology $H^{0,1}(X,\mathbb{C})$. Then we take $\alpha^{0,1}=v\Lambda\Theta_{(A_0,\phi_0)}$ for an element $v\in V$, thus $\alpha=2\textrm{Re} (v)\Lambda\Theta_{(A_0,\phi_0)}$, and  $\beta=c\phi_0$ for a nonzero constant $c$.  It follows that  $(\alpha^{0,1}, \beta)$ forms a holomorphic deformation pair from $(A_0,\phi_0)$ being a strong Yang-Mills-Higgs pair. Thereby we have
\begin{align*}
  Q_{(A_0,\phi_0)}(\alpha,\tilde\beta)=\langle\alpha,\Delta_{A_0}\alpha\rangle-4|c|^2||[\phi_0,\phi_0^*]||^2.
\end{align*}
Since $\Lambda[\phi_0,\phi_0^*]\neq0$ means $[\phi_0,\phi_0^*]\neq0$ and $c$ can be chosen to be sufficiently large such that $Q_{(A_0,\phi_0)}(\alpha,\tilde\beta)<0$, $H^{0,1}(X,\mathbb{C})$ has to vanish, thus $H^1(X,\mathbb{C})$ must also vanish by Hodge decomposition theorem.
\end{proof}
\paragraph{\textbf{Acknowledgements}} The authors would like to thank Prof. Kang Zuo and Prof. Xi Zhang for their supports.

 \end{document}